\documentclass{article}

     \usepackage[preprint]{neurips_2019}

\usepackage[utf8]{inputenc} 
\usepackage[T1]{fontenc}    
\usepackage{hyperref}       
\usepackage{url}            
\usepackage{booktabs}       
\usepackage{amsfonts}       
\usepackage{nicefrac}       
\usepackage{microtype}      

\usepackage{graphicx}
\usepackage{lscape}
\usepackage{setspace}
\usepackage{subcaption}
\usepackage{color}
\usepackage{sgame}
\usepackage{wrapfig}
\usepackage{sgamevar}
\usepackage{amsmath,amssymb,amsthm,cancel,staves}
\usepackage{mathtools}
\usepackage{enumitem}
\usepackage{titlesec}
\usepackage{tikz}
\usepackage{breakcites}
\usepackage{indentfirst}
\usetikzlibrary{arrows}
\usepackage{accents}

\DeclareMathOperator*{\argmin}{arg\,min}
\newcommand{\ATT}{\textrm{ATT}}

\newcommand{\ATUT}{\textrm{ATUT}}

\newcommand{\Tavg}{\frac{1}{N_1} \sum_{i:T_i=1}}
\newcommand{\Cavg}{\frac{1}{N_0} \sum_{i:T_i=0}}

\newtheorem{theorem}{Theorem}[section]

\newtheorem{assumption}{Assumption}

\theoremstyle{definition}

\numberwithin{remark}{section}
\numberwithin{equation}{section}
\numberwithin{theorem}{section}

\renewcommand{\(}{\left(}
\renewcommand{\)}{\right)}
\renewcommand{\hat}{\widehat}

\renewcommand{\hat}{\widehat}
\renewcommand{\leq}{\leqslant}
\renewcommand{\geq}{\geqslant}

\DeclareMathOperator*{\argminA}{arg\,min}

\renewcommand{\(}{\left(}
\renewcommand{\)}{\right)}
\renewcommand{\[}{\left[}
\renewcommand{\]}{\right]}

\title{Matching on What Matters:\\
	{\Large A Pseudo-Metric Learning Approach to Matching Estimation in High Dimensions} 
	}

\author{%
  Gentry Johnson \\
  Department of Economics\\
  University of Maryland\\
  College Park, MD 20742 \\
  \texttt{johnsong@econ.umd.edu} \\
  \And
  Brian Quistorff \\
  AI + Research \\
  Microsoft \\
  Redmond, WA 98052 \\
  \texttt{Brian.Quistorff@microsoft.com} \\
  \AND
  Matt Goldman \\
  AI + Research \\
  Microsoft \\
  Redmond, WA 98052 \\
  \texttt{mattgoldman5850@gmail.com} \\
}

\begin{document}

\maketitle

\begin{abstract} 
When pre-processing observational data via matching, we seek to approximate each unit with maximally similar peers that had an alternative treatment status--essentially replicating a randomized block design. However, as one considers a growing number of continuous features, a curse of dimensionality applies making asymptotically valid inference impossible \citep{abadie2006large}. The alternative of ignoring plausibly relevant features is certainly no better, and the resulting trade-off substantially limits the application of matching methods to ``wide'' datasets. Instead, \cite{li2017matching} recasts the problem of matching in a metric learning framework that maps features to a low-dimensional space that facilitates ``closer matches'' while still capturing important aspects of unit-level heterogeneity. However, that method lacks key theoretical guarantees and can produce inconsistent estimates in cases of heterogeneous treatment effects. Motivated by straightforward extension of existing results in the matching literature, we present alternative techniques that learn latent matching features through either MLPs or through siamese neural networks trained on a carefully selected loss function. We benchmark the resulting alternative methods in simulations as well as against two experimental data sets--including the canonical NSW worker training program data set--and find superior performance of the neural-net-based methods.
\end{abstract}

\section{Introduction}
\label{sec:intro}
While experimentation approaches to causal inference have strong theoretical guarantees, they are impractical in some settings. This leads researchers to rely on observational methods which suppose that conditioning on available covariates is sufficient to yield unbiased estimates. This can be especially tenuous (or prone to researcher manipulation) when the conditioning variables are selected in an \textit{ad hoc} way. Recognizing this limitation, there has been growing interest in importing machine learning (ML) tools to automate and accelerate the model selection process \citep{athey2017beyond}. Most examples of this work \citep{chernozhukov2016double,wager2018estimation} rely on ``black-box'' ML algorithms that may obscure the underlying counterfactual reasoning. By contrast, matching methods yield a set of direct unit-to-unit comparisons that are fully transparent to policy makers, but are much less robust to the curse of dimensionality presented by high-dimensional data. Nonetheless, there has been relatively little work on the problem of model selection in high-dimensions for these methods.

Instead, propensity score matching (PSM) remains the most commonly used matching method used in applied economics and related fields. This approach side-steps the curse of dimensionality by collapsing all covariates to a single number reflecting probability of treatment \citep{McCaffrey2004,Wyss2014}. However, PSM has been criticized because it does not guarantee covariate balance in the resulting sample, thus allowing significant model dependence and inefficiency in downstream estimations \citep{King2016}. Prognostic score matching \citep{Hansen2008} also has the advantage of collapsing the covariates used for matching to a single dimension. While it does not explicitly balance covariates, it does balance on \textit{counterfactual outcomes}, eliminating the potentially large efficiency losses associated with PSM.

Matching directly on the space of covariates--using methods like Mahalanobos Distance Matching (MDM) or Coarsened Exact Matching (CEM)--has gained popularity in fields such as Statistics, Epidemiology, Sociology, and Political Science \citep{ho2007matching}. This has the advantage of not requiring correct specification of either a propensity or prognostic model. Additionally, these methods provide better covariate balance and, to the extent those covariates drive variance in outcomes,  improve efficiency of downstream estimations when compared to PSM. Finally, these methods prevent the creation of \textit{unintuitive matches} that happen when units with very different covariate values have similar estimated propensity (or prognostic) scores.
However, MDM and CEM are particularly vulnerable to the curse of dimensionality discussed in \cite{abadie2006large} and demonstrated in \cite{Gu1993}. That is, even for the case of two continuous covariates, the non-parametric bias resulting from covariate differences between matched pairs can be large enough to make asymptotically valid inference impossible.
Along with \cite{li2017matching}, we explore a middle ground to recast the problem of matching in a metric learning framework that maps our covariates to a low-dimensional space that facilitates ``closer matches'' while still capturing important aspects of unit-level heterogeneity. These methods are a natural choice as they were initially designed to predict similarity in high-dimensional problems like facial recognition \citep{parkhi2015deep}, person re-identification \citep{liao2015person}, or image retrieval \citep{hoi2010semi}.

In order to proceed, we first establish very high-level conditions on the behavior of a learned metric in order to satisfy the consistency of matching estimators. Then, we utilize feed-forward neural networks (NN) and siamese neural networks (SNN) trained (separately) on outcomes and treatment labels and extract the semi-final layer of these networks as low-dimensional representations that are, nonetheless, sufficient to control for the confounding effects of our covariates. The dimension of this embedding layer can be easily altered to trade off a richer representation for concerns about non-parametric bias. Through simulation and comparison to two canonical experimental datasets from the applied social sciences, we show our methods perform better than available competitors.



%

\section{Theoretical framework}
\label{sec:econometric_framework}

Following the potential outcome framework, we let $T \in \{0,1\}$ denote a unit's treatment status and $Y(1), \ Y(0)$ denote the outcomes that would be realized for each unit if it were treated (or not). Additionally, we adopt the following conventional assumptions.

\begin{assumption}(compact support)
	\label{ass:1}
	Let $X$ be a random vector of continuous covariates with dimension $k$ and with density bounded away from zero on some compact and convex support $\mathbb{X}$.
\end{assumption}

\begin{assumption}
	\label{ass:2}
	For some fixed $\eta>0$ and $\forall x \in \mathbb{X}:$
	\begin{itemize}
		\item (unconfoundedness) $T \perp (Y(0), Y(1)) | X$
		\item (common support) $\rho(x) = P(T=1|X) \in (\eta,1-\eta)$
		\item (continuity) Denote
		\begin{align*}
		m(x) &:=E[Y(0)|X=x]\\
		\tau(x) &:= E[Y(1)|X=x] - E[Y(0)|X=x],
		\end{align*}
		then $\rho$, $m$, $\tau$ are all continuous functions of $x$.
	\end{itemize}
\end{assumption}

\begin{assumption}(iid data)
	\label{ass:3}
	$\{Y_i,T_i,X_i\}$ are independent and identically distributed draws from some joint distribution of $(Y,T,X)$,
\end{assumption}
Implicit in Assumption~\ref{ass:3} is the Stable Unit Treatment Value Assumption (SUTVA) that rules out any contaminating effects of one unit's treatment onto another unit's outcomes.

The average treatment effect on the treated (ATT) or untreated (ATUT) populations are then estimated under nearest neighbor matching (NNM) as
\begin{align}
\label{eqn:att_def}
\hat \ATT_d &= \Tavg \left[Y_i - Y_{j^*_d(i)}\right] \\
\notag
\hat \ATUT_d &= \Cavg \left[Y_i - Y_{j^*_d(i)}\right],
\end{align}
where $j^*_d(i) = \argmin_{k:T_i \neq T_k} d(X_i,X_k)$ gives the nearest neighbor with opposing treatment status. We write the $d$ subscript to emphasize the estimator's dependence on some (potentially learned) distance metric in our setting. Now we establish high-level conditions that $d$ must obey in order to satisfy consistency of the resulting estimator.

\begin{theorem} [Consistency]
	\label{thm:consistency}
	Suppose a metric $d$ on the space of $\mathbb X$. Additionally,
	\begin{enumerate}[label=\alph*.]
		\item suppose $\exists \ C >0$ such that for every $x,y \in \mathbb{X}$ \textit{either}
		\begin{itemize}
			\item $PSM$: $d(x,y) \geq C \cdot |\rho(x) - \rho(y)|$
			\item $PGM_C$: $d(x,y) \geq C \cdot |m(x) - m(y)|$
		\end{itemize}
		then $\hat{\ATT}_{d}$ is an asymptotically unbiased estimator of $\ATT$.
		\item suppose $\exists \ C >0$ such that for every $x,y \in \mathbb{X}$ \textit{either}
		\begin{itemize}
			\item $PSM$: $d(x,y) \geq C \cdot |\rho(x) - \rho(y)|$
			\item $PGM_T$: $d(x,y) \geq C \cdot |(m(x) + \tau(x)) - (m(y) + \tau(y))|$
		\end{itemize}
		then $\hat{\ATUT}_{d}$ is an asymptotically unbiased estimator estimator of $\ATUT$.
	\end{enumerate}
And consistency of either estimator would follow from application of a law of large numbers to the residuals.
\end{theorem}
\begin{proof}
	See proof in Supplement D
\end{proof}

Theorem~\ref{thm:consistency} guarantees consistency for metrics that suitably penalize either deviations in propensity score or a specific prognostic score (depending on the estimand). It is immediate to see that it can be applied to the already well-established results for MDM matching as well as propensity score and prognostic score matching for the cases when those models are correctly specified. Furthermore, it demonstrates the consistency of methods which (given guarantees on convergence) learn a matching metric that utilizes outcomes as labels, provided that only outcomes for control units are used to learn a metric for estimating the ATT, or conversely, that only outcomes for treated units are used to learn a metric for estimating the ATUT.\footnote{However be combined with additional results about convergence of the underlying metric learning methods in order to insure consistency in practice.} This closely parallels the results in \cite{Antonelli2016} that require two separate prognostic scores for average treatment effect estimation.

Critically, it does not support the consistency of matching estimators that utilize metrics learned on the pooled samples of outcome labels, including the procedure detailed in \cite{li2017matching}. To demonstrate the potential failure of consistency in such settings, we pose a simple example with just a single covariate ($x\sim U[-1,1]$) and outcomes given by
\begin{align}
\label{eqn:example}
E[Y(0)|x] &= \frac{x}{2} \qquad E[Y(1)|x] = 2 \qquad  \tau(x) =  E[Y(1)|x] -  E[Y(0)|x] = \frac{4-x}{2}\\
\intertext{and treatment assigned with propensity scores}
\notag
\rho(x) &= 
\left\{
\begin{array}{ll}
\frac{x}{4-x}  & \mbox{if } x \geq 0, \\
\frac{-3x}{4-x} & \mbox{if } x < 0.
\end{array}
\right.
\end{align}
A metric learned to minimize mean squared error on the pooled (treatment and control) outcome data would converge to
\begin{align*}
d_0(x_1,x_2) &= (E[Y|x_1] - E[Y|x_2])^2=(|x_1| - |x_2|)^2.
\end{align*}
That is, for a treated unit with $x_i=\frac{1}{2}$ we would be indifferent between matching to a control unit with $x=\frac{1}{2}$, implying $\hat{\tau}_i=\frac{1}{2}$, and matching to a control unit with $x=-\frac{1}{2}$, implying $\hat{\tau}_i=\frac{3}{2}$. It is immediately evident that such an estimator would be inconsistent for both the ATT and ATUT parameters. The example in \eqref{eqn:example} is highly stylized, but as long as one supposes heterogeneous treatment effects, it will be straightforward to construct counter-examples to the consistency of matching metrics learned on pooled outcome data.

In addition to using a control-only or treatment-only subsample of training units, an alternative but somewhat more parametric attempt to mitigate potential bias resulting from heterogenous treatment effects is detailed in \cite{JohanssonEtAl_icml16}, wherein treatment status is concatenated with the learned data representation only at the stage immediately before prediction.

\section{Proposed Matching Framework and Methods} \label{sec:estimators}

In this section, we provide a description of a general procedure to pre-process data for matching estimation. This procedure is then combined with NNM to produce a full matching method. We describe in additional detail how this procedure can be applied to the pseudo-metric learning paradigm specifically.

\subsection{Approach Overview}

Our general framework, which is conformable to a variety of pre-processing techniques, involves training a predictive algorithm to predict an outcome $Y$, as well as training a separate instance of the algorithm to predict treatment status $D$. We then combine the learned representation of the data from these two tasks--where the method of combination depends on the predictive algorithm used--and carry out nearest-neighbor matching on the resulting object. Computationally, the complexity of the matching process is dominated by the component ML models. As these are standard model types their complexity is well established in the literature.

\subsection{Pseudo-metric learning methods}

The following methods employ different neural net architectures in order to discover some low-dimensional representation of the original covariates on which a standard distance metric (e.g. Euclidean) can be applied to, together, constitute a learned pseudo-metric. In analogy to the alternatives of propensity and prognostic score matching, we learn two separate low-dimensional representations of the data, one in which units with similar potential outcomes are nearby each other,\footnote{When training the neural nets which learn similarity in outcomes, we use control-only or treatment-only subsamples. This decision is motivated by the previous discussion in Section \ref{sec:econometric_framework} regarding heterogeneous treatment effects.} and one in which units with similar propensity for treatment are nearby each other. 

For a given unit $i$, these separate representations are characterized by two vectors of continuous-valued features, $\mathbf{m_{i,y}}$ and $\mathbf{m_{i,d}}$. We propose performing matching on the space defined by the union of these two representations, $M_{union} = [M_{y} \ \  M_{d}]$, where $dim(M_{d}) = n \times z_d$ and $dim(M_{y}) = n \times z_y$. It is possible that $z_d = z_y$ if the researcher so chooses, but this need not be true. Note that unlike methods which aim to select some subset of the original covariates, the combination of the learned features discussed here will include all learned features resulting from estimating the separate treatment and outcome models. Thus, the researcher, through specifying the dimension of the hidden layers to be extracted and used in matching, can explicitly control the dimension of the matching space. As a practical matter, we drop learned features which are perfectly correlated with another learned feature and learned features with near-zero variance. We recommend standard hyperparameter tuning for both the number of hidden layers as well as their size.

\subsubsection{NN method}

NNs are, theoretically, universal function approximators \citep{Cybenko1989}. The success of NNs derives from the capacity of a network's inner layers to learn successive transformations of the data. Extracting these inner layers and using them as input into a second model is a known technique for dimensionality reduction \citep{Hinton504}.

Specifically, we will refer to the final hidden layers of the network that predict treatment and outcome as $M_{d}$ and $M_{y}$, respectively. These are of dimension $z_d$ and $z_y$ respectively--parameters that the researcher can control when designing the network architecture that will determine the dimension of the matching space. Additionally, denote the vectors of final hidden layer weights as $\sigma_{l}^{d}$ and $\sigma_{l}^y$, respectively, and let $e_{n} = [1 \ 1 \ 1 \ . . . \ 1]$ be a $(1 \times n)$ vectors of ones. Then, letting $\otimes$ be the Kronecker product operator, define $A$, an $(n \times z_d)$ matrix, and $B$, an $(n \times z_y)$ matrix as

\begin{equation}
A = e' \otimes \sigma_{l}^{d}, \ \ B= e' \otimes \sigma_{l}^{y}
\end{equation}

Using the above definitions, and letting the operator $\odot$ refer to element-wise matrix multiplication, we propose matching on

\begin{equation}
X_{s} = [M_{d} \odot A \ \ M_{y} \odot B]_{n \times (z_y + z_d)}
\end{equation}

The matrix $X_{s}$ can be thought of simply as the matrix of scaled learned features. If a feature in the final hidden layer receives a very large weight as it enters the output layer, it is desirable to also increase the scale of that feature in the downstream matching estimation. Similarly, if a feature in the final hidden layer receives a near-zero weight as it is passed to the output layer, we would like to assign less importance to matching closely on that feature than on others.

\subsubsection{SNN method}

The siamese neural network structure is not designed to predict a value, but rather to learn a mapping to a low-dimensional manifold such that alike observations are nearby each other and dissimilar observations are far from each other on the manifold. Each pass through the network involves two observations, $X_{i}$ and $X_{j}$. Each observation is passed through identical networks, terminating in a layer whose length the researcher controls. 

The distance between the two layers is then input into one of two loss functions. The SNN which is designed to learn similarity in outcomes uses the loss function described in Equation (2.2). The SNN which is designed to learn similarity in treatment propensity uses a standard contrastive loss function as in Equation (2.3).

We recommend extracting the final hidden layer of each SNN, and matching on the combination of these layers. Denote the layer extracted from the treatment-target SNN as $M_{d}$ and the layer extracted from the outcome-target SNN as $M_{y}$. In the SNN architecture there is no set of weights connecting the extracted layer to some terminal node, instead the extracted layer enters directly into the loss function (2.2) or (2.3). Thus, unlike in the NN Method, there is no need to weight the extracted layer before its input into the matching estimator. We propose matching on

\begin{equation}
X_{s} = [M_{d} \ \ M_{y}]_{n \times (z_y + z_d)}
\end{equation}

The rows of the matrix $X_{s}$ can be seen as representing the position of each observation $i$ on the learned manifold, where the mapping to that manifold has been explicitly constructed such that alike pairs are nearby each other in terms of Euclidean distance. This structure is naturally amenable to a standard matching procedure which assumes that observations that are nearby each other in terms of of Euclidean distance are alike. Indeed, this transformation of the raw covariate space renders that implicit assumption in the standard matching procedure far more plausible.

\section{Simulations}
\label{sec:simulations}


In order to benchmark the performance of the NN and SNN methods proposed here, we compare them to a number of other methods described in Table~\ref{methods}--including variable selection methods that choose a subset of the raw features---across three different simulated environments characterized by DGPs of increasing complexity.

\begin{table}[ht]
	\centering
	\caption{Matching methods used in simulations}\label{methods}
	\begin{tabular}{lcc}
		\hline
		Abbrev. & Type & Matching Features \\\hline
		NN & Pseudo-metric & Union of semi-final layers from NNs trained on $Y,D$ \\
		SNN & Pseudo-metric & Union of semi-final layers from SNNs trained on $Y,D$ \\
		$L_1$$^*$ & Var. Selection &Union of features selected by Lasso models on $Y,D$ \\
		RRF$^*$ & Var. Selection & Union of features selected by RRFs on $Y,D$\\
		PSM& Prop. Score & Propensity scores from logit regression of $D$ onto $X$\\
		PSMSQ& Prop. Score &Propensity scores from logit regression of $D$ onto $\{X, X^2\}$\\
		U. Oracle& Var. Selection & Union of features known to impact \textit{either} $Y$ \textit{or} $D$\\
		Int. Oracle& Var. Selection &Intersection of features known to impact \textit{both} $Y$ \textit{and} $D$\\\hline
		\end{tabular}
\end{table}

For all simulated environments, the following are true:

\begin{equation}
X_{i} \sim N(0_K, I_K); \ K = 50 \ N = 8000 ; \beta_{0} = 1
\end{equation}

Where $X_{i}$ denotes the vector of covariates for observation $i$, $K$ is the dimension of the covariate space, $N$ is the number of observations in a given simulation, and $\beta_0$ is the true treatment effect. We examine performance under three different DGPs: a sparse linear DGP, a sparse linear DGP with quadratic terms, and a NN DGP. Complete descriptions of the simulated environment can be found in Supplement A.

As additional benchmarks, we include two methods which generate a matching space through variable selection. First, in accordance with Section \ref{sec:estimators}, we estimate a pair of LASSO regressions to choose matching features. In the second alternative, we estimate a pair of Regularized Random Forests (RRF) \citep{Deng2012} to choose matching features on the basis of variable importance. Supplement B contains a more complete treatment of these alternatives.

Examining Table \ref{sparse_linear} it is evident that, when the DGP is relatively simple, both PSM and PSMSQ do quite well in RMSE terms. Importantly, the two neural-net-based methods we propose, NN and SNN, perform nearly identically to PSM and PSMSQ in this setting.

At a first pass, it may seem counter-intuitive that oracle estimators--endowed with perfect knowledge of the relevant covariates--will underperform supervised algorithms which seek to recover some form of this knowledge. However, the non-parametric bias induced by matching to a nearest neighbor is $O(n^{-1/K})$ for problems with $K$ continuous covariates \citep{abadie2006large}. This explains the poor performance of the oracle estimators and further motivates the pseudo-metric learning paradigm.

\begin{table}[ht]
\centering
\caption{Sparse linear DGP simulation results\label{sparse_linear}}
\begin{tabular}{rrrrrrrrrr}
  \hline
 & NN & SNN & $L_1$ & RRF & PSM & PSMSQ & U. Oracle & Int. Oracle \\ 
  \hline
Mean & 0.96 & 1.00 & 1.42 & 1.52 & 0.99 & 0.99 & 1.29 & 1.13 \\ 
  SD & 0.02 & 0.04 & 0.04 & 0.03 & 0.05 & 0.04 & 0.03 & 0.05 \\ 
  RMSE & 0.05 & 0.04 & 0.42 & 0.52 & 0.05 & 0.04 & 0.30 & 0.14 \\ 
   \hline
\end{tabular}

\end{table}


In Table \ref{sparse_linear_sq} we see that, as the DGP grows slightly non-linear, the performance of PSM and PSMSQ begins to deteriorate. It should be expected that PSM, now explicitly misspecified, should begin to falter. However, PSMSQ sees a non-trivial uptick in RMSE as well.

\begin{table}[ht]
\centering
\caption{Sparse linear with sq. terms DGP simulation results}\label{sparse_linear_sq}
\begin{tabular}{rrrrrrrrrr}
  \hline
 & NN & SNN & $L_1$ & RRF & PSM & PSMSQ & U. Oracle & Int. Oracle \\ 
  \hline
Mean & 0.99 & 1.04 & 1.81 & 1.92 & 1.19 & 1.01 & 1.66 & 1.34 \\ 
  SD & 0.05 & 0.07 & 0.16 & 0.17 & 0.26 & 0.13 & 0.13 & 0.11 \\ 
  RMSE & 0.05 & 0.08 & 0.83 & 0.94 & 0.32 & 0.13 & 0.67 & 0.36 \\ 
   \hline
\end{tabular}

\end{table}


The simulation results from the NN DGP in Table \ref{ann} indicate that, as the DGP grows increasingly complex, PSM and PSMSQ become wholly unusable. SNN and NN exhibit a very gentle uptick in RMSE, but are still nearly unbiased. If the researcher takes seriously the possibility that the DGP may be highly non-linear, then it is prudent to employ something akin to NN or SNN--it is harmless if the DGP turns out to be simple, but provides significant gains if is not.
\begin{table}[ht]
\centering
\caption{NN DGP simulation results}\label{ann}
\begin{tabular}{rrrrrrrr}
  \hline
 & NN & SNN & $L_1$ & RRF & PSM & PSMSQ \\ 
  \hline
Mean & 1.01 & 1.07 & 1.58 & 1.56 & 1.51 & 1.39 \\ 
  SD & 0.10 & 0.21 & 0.06 & 0.04 & 0.60 & 0.46 \\ 
  RMSE & 0.10 & 0.21 & 0.58 & 0.56 & 0.77 & 0.59 \\ 
   \hline
\end{tabular}

\end{table}

\section{Testing performance with experimental data} \label{sec:application}
\subsection{LaLonde (1986) dataset}

We first examine the performance of the methods proposed in Section \ref{sec:estimators} on the canonical \cite{lalonde1986evaluating} \href{http://users.nber.org/~rdehejia/data/nswdata2.html}{data set}. Following \cite{imbens2014}, we look specifically at the experimental and non-experimental versions of the original data set introduced in \cite{Dehejia1999}.

In the LaLonde data, applicants to National Supported Work (NSW), a labor market training program, were selected at random for participation. Using the this data alone, we can learn the experimental ATT. Combining these data with a non-experimental comparison group derived from the Panel Study of Income Dynamics (PSID), it is possible to benchmark the performance of various causal inference methods designed to recover the ATT from observational data.

The variable of interest is labor market earnings in 1978, and available features include earnings in 1974 and 1975 as well as various demographic variables. In Table \ref{tab:lalonde_results}, we compare non-experimental results with the experimentally estimated ATT of \$1,794. The NN and SNN estimates, standard errors, and confidence intervals are taken from the average over 100 iterations of each method. In the case of the non-experimental LaLonde data, the control and treated groups are significantly different in the observed features. For that reason, we find find it prudent to average over different iterations of the neural-net based methods in a situation with such extreme baseline control-treated feature imbalance.

As is standard practice in cases of significant class imbalance, we recommend oversampling in NN and SNN training if the treated units comprise less than $10\%$ of the entire sample. Other methods of protecting against overfit may also be used as well.

\begin{table}[ht]
\centering
\caption{LaLonde non-experimental ATT results} \label{tab:lalonde_results}
\begingroup\small
\begin{tabular}{rrrrr}
  \hline
 & Est.* & Difference & SE** & 95\% CI \\ 
  \hline
\textbf{Experimental} & \textbf{1794.34} &  \textbf{0.00} &  \textbf{671.00} & \textbf{(479, 3110)} \\ 
  NN & 1632.74 & -161.60 &  872.79 & (-78, 3343) \\ 
  SNN & 1736.51 &  -57.84 &  795.11 & (178, 3295) \\ 
  PSM &  897.94 & -896.40 & 1045.56 & (-1151, 2947) \\ 
  PSMSQ & 2109.21 &  314.87 &  990.90 & (167, 4051) \\ 
  OLS &  914.65 & -879.69 &  551.32 & (-166, 1995) \\ 
   \hline\multicolumn{5}{l}{*Estimates for NN and SNN methods are averaged over 100 runs each.} \\
  \multicolumn{5}{l}{**SE are calculated as in Abadie and Imbens (2006) for matching estimators.}
\end{tabular}

\endgroup
\end{table}

Standard errors are calculated according to \cite{abadie2006large} for all estimators besides OLS. Among all methods, SNN comes the closest to recovering the experimental treatment effect, with NN outperforming all non-SNN methods by a considerable margin as well. While OLS provides the lowest standard error, it is considerably biased.\footnote{We attempted to benchmark our results on this data set against those reported in \cite{li2017matching}, but were unable to recreate their exact environment.}

\subsection{IHDP (1993) dataset}

We additionally examine the performance of the methods proposed in Section \ref{sec:estimators} on a modified version of the IHDP (Infant Health and Development Program) \href{https://www.icpsr.umich.edu/icpsrweb/HMCA/studies/9795}{data set}. The IHDP data set originates from a randomized longitudinal trial designed to discover the effect of a comprehensive set of interventions designed to curb health and developmental issues in low birth weight, premature infants.

In order to simulate an observational study, we follow \cite{hill2011}, \cite{li2017matching}, and others in removing all children with non-white mothers from the treatment group. The subset of removed infants is clearly non-random and therefore will induce the type of structural feature imbalance that we typically expect in non-experimental data.

Because the theoretical framework we present maintains the unconfoundedness assumption, we simulate outcomes using only pretreatment features and treatment assignment. The response surfaces we use follow \cite{li2017matching} exactly. A full description can be found in Supplement C.

We simulate 50 response surfaces and report the results in Table \ref{tab:ihdp_results}. Both NN and SNN result in small bias and RMSE, providing further support for the efficacy of these pseudo-metric learning methods across a variety of settings. While \cite{li2017matching} only report the absolute value of the error in average treatment effect, they report an error of $0.16$, which is bested by both the NN and SNN methods. \newline

\begin{table}[ht]
\centering
\caption{IHDP simulated outcome results (ATT = 4)} \label{tab:ihdp_results}
\begingroup\small
\begin{tabular}{rrrr}
  \hline
 & Avg Est.* & Avg Difference* & RMSE**  \\ 
  \hline
NN & 3.89 & -0.11 & 0.28 \\ 
  SNN & 4.13 & 0.13 & 0.34 \\ 
  PSM & 3.74 & -0.26 & 1.03 \\ 
  PSMSQ & 3.92 & -0.08 & 0.46 \\ 
  OLS & 3.83 & -0.17 & 0.33 \\ 
   \hline\multicolumn{4}{l}{*Averages over 50 sets of simulated outcomes.} \\
  \multicolumn{4}{l}{**RMSE is calculated using results from 50 simulated outcomes.}
\end{tabular}

\endgroup
\end{table}

\section{Conclusion}
\label{sec:conclusion}
When estimating effects of conditionally exogenous treatment, pre-processing data with matching is a popular tool, and researchers with high-dimensional data may face a difficult decision about how best to define an appropriate metric. Following \cite{li2017matching}, we recast the problem of matching in a metric learning framework that maps our covariates to a low-dimensional space. This facilitates ``closer matches'' while still capturing important aspects of unit-level heterogeneity. 

We provide general conditions under which a learned distance metric is sure to lead to consistent estimation. This illuminates the importance of using control-only or treatment-only data when using supervised learning to discover a metric from outcome data. We also provide applied researchers with two new methods that leverage both MLPs and siamese neural nets and which compare favorably to state-of-the-art alternatives.


\small
\bibliographystyle{apalike} 

\begin{thebibliography}{}

\bibitem[Abadie and Imbens, 2006]{abadie2006large}
Abadie, A. and Imbens, G.~W. (2006).
\newblock Large sample properties of matching estimators for average treatment
  effects.
\newblock {\em econometrica}, 74(1):235--267.

\bibitem[Abadie and Imbens, 2016]{abadie2016matching}
Abadie, A. and Imbens, G.~W. (2016).
\newblock Matching on the estimated propensity score.
\newblock {\em Econometrica}, 84(2):781--807.

\bibitem[Antonelli et~al., 2016]{Antonelli2016}
Antonelli, J., Cefalu, M., Palmer, N., and Agniel, D. (2016).
\newblock Doubly robust matching estimators for high dimensional confounding
  adjustment.

\bibitem[Athey, 2017]{athey2017beyond}
Athey, S. (2017).
\newblock Beyond prediction: Using big data for policy problems.
\newblock {\em Science}, 355(6324):483--485.

\bibitem[Chernozhukov et~al., 2016]{chernozhukov2016double}
Chernozhukov, V., Chetverikov, D., Demirer, M., Duflo, E., Hansen, C., and
  Newey, W.~K. (2016).
\newblock Double machine learning for treatment and causal parameters.
\newblock Technical report, cemmap working paper, Centre for Microdata Methods
  and Practice.

\bibitem[Cybenko, 1989]{Cybenko1989}
Cybenko, G. (1989).
\newblock Approximation by superpositions of a sigmoidal function.
\newblock {\em Mathematics of Control, Signals and Systems}, 2(4):303--314.

\bibitem[Dehejia and Wahba, 1999]{Dehejia1999}
Dehejia, R.~H. and Wahba, S. (1999).
\newblock Causal effects in nonexperimental studies: Reevaluating the
  evaluation of training programs.
\newblock {\em Journal of the American Statistical Association},
  94(448):1053--1062.

\bibitem[Deng and Runger, 2012]{Deng2012}
Deng, H. and Runger, G. (2012).
\newblock Feature selection via regularized trees.
\newblock In {\em The 2012 International Joint Conference on Neural Networks
  (IJCNN)}, pages 1--8.

\bibitem[Gu and Rosenbaum, 1993]{Gu1993}
Gu, X.~S. and Rosenbaum, P.~R. (1993).
\newblock Comparison of multivariate matching methods: Structures, distances,
  and algorithms.
\newblock {\em Journal of Computational and Graphical Statistics},
  2(4):405--420.

\bibitem[Hansen, 2008]{Hansen2008}
Hansen, B.~B. (2008).
\newblock The prognostic analogue of the propensity score.
\newblock {\em Biometrika}, 95(2):481--488.

\bibitem[Hill, 2011]{hill2011}
Hill, J.~L. (2011).
\newblock Bayesian nonparametric modeling for causal inference.
\newblock {\em Journal of Computational and Graphical Statistics},
  20(1):217--240.

\bibitem[Hinton and Salakhutdinov, 2006]{Hinton504}
Hinton, G.~E. and Salakhutdinov, R.~R. (2006).
\newblock Reducing the dimensionality of data with neural networks.
\newblock {\em Science}, 313(5786):504--507.

\bibitem[Ho et~al., 2007]{ho2007matching}
Ho, D.~E., Imai, K., King, G., and Stuart, E.~A. (2007).
\newblock Matching as nonparametric preprocessing for reducing model dependence
  in parametric causal inference.
\newblock {\em Political analysis}, 15(3):199--236.

\bibitem[Hoi et~al., 2010]{hoi2010semi}
Hoi, S.~C., Liu, W., and Chang, S.-F. (2010).
\newblock Semi-supervised distance metric learning for collaborative image
  retrieval and clustering.
\newblock {\em ACM Transactions on Multimedia Computing, Communications, and
  Applications (TOMM)}, 6(3):18.

\bibitem[Imbens, 2015]{imbens2014}
Imbens, G.~W. (2015).
\newblock Matching methods in practice.
\newblock {\em Journal of Human Resources}, 50(2):373--419.

\bibitem[Johansson et~al., 2016]{JohanssonEtAl_icml16}
Johansson, F., Shalit, U., and Sontag, D. (2016).
\newblock Learning representations for counterfactual inference.
\newblock In Balcan, M.~F. and Weinberger, K.~Q., editors, {\em Proceedings of
  The 33rd International Conference on Machine Learning}, volume~48 of {\em
  Proceedings of Machine Learning Research}, pages 3020--3029, New York, New
  York, USA. PMLR.

\bibitem[King and Nielsen, 2016]{King2016}
King, G. and Nielsen, R. (2016).
\newblock Why propensity scores should not be used for matching.
\newblock Mimeo.

\bibitem[LaLonde, 1986]{lalonde1986evaluating}
LaLonde, R.~J. (1986).
\newblock Evaluating the econometric evaluations of training programs with
  experimental data.
\newblock {\em The American economic review}, pages 604--620.

\bibitem[Li and Fu, 2017]{li2017matching}
Li, S. and Fu, Y. (2017).
\newblock Matching on balanced nonlinear representations for treatment effects
  estimation.
\newblock In {\em Advances in Neural Information Processing Systems}, pages
  929--939.

\bibitem[Liao et~al., 2015]{liao2015person}
Liao, S., Hu, Y., Zhu, X., and Li, S.~Z. (2015).
\newblock Person re-identification by local maximal occurrence representation
  and metric learning.
\newblock In {\em Proceedings of the IEEE conference on computer vision and
  pattern recognition}, pages 2197--2206.

\bibitem[McCaffrey et~al., 2004]{McCaffrey2004}
McCaffrey, D.~F., Ridgeway, G., and Morral, A.~R. (2004).
\newblock Propensity score estimation with boosted regression for evaluating
  causal effects in observational studies.
\newblock {\em Psychological Methods}, 9(4):403--425.

\bibitem[Parkhi et~al., 2015]{parkhi2015deep}
Parkhi, O.~M., Vedaldi, A., Zisserman, A., et~al. (2015).
\newblock Deep face recognition.
\newblock In {\em BMVC}, volume~1, page~6.

\bibitem[Wager and Athey, 2018]{wager2018estimation}
Wager, S. and Athey, S. (2018).
\newblock Estimation and inference of heterogeneous treatment effects using
  random forests.
\newblock {\em Journal of the American Statistical Association},
  113(523):1228--1242.

\bibitem[Wyss et~al., 2014]{Wyss2014}
Wyss, R., Ellis, A.~R., Brookhart, M.~A., Girman, C.~J., Funk, M.~J., LoCasale,
  R., and St{\"u}rmer, T. (2014).
\newblock The role of prediction modeling in propensity score estimation: An
  evaluation of logistic regression, {bCART}, and the covariate-balancing
  propensity score.
\newblock {\em American Journal of Epidemiology}, 180(6):645--655.

\end{thebibliography}

\nocite{abadie2016matching}
\newpage
\appendix

\part*{Supplements}

\section{Simulation settings}
Descriptions for each simulated environment described in Section 4 are provided below. Each DGP is simulated 1000 times.

\begin{itemize}

	\item[1.] A sparse linear DGP characterized by
	
	\begin{equation}
	Y_{i} = D_{i}\beta_{0} + X_{i} \boldsymbol{\gamma} + \epsilon_{i} ; \ \epsilon_{i} \sim N(0,1)
	\end{equation}
	
	\begin{equation}
	D_{i} = Bernoulli  \left( \frac{1}{1 + e^{-(X_{i} \boldsymbol{\omega})}} \right)
	\end{equation}

	Additionally define $L_{nz} = \{\omega_{k} \in \boldsymbol{\omega} \vert \omega_{k} \neq 0\}$ and $G_{nz}  = \{\gamma_{k} \in \boldsymbol{\gamma} \vert \gamma_{k} \neq 0\}$. Then,
	
	\begin{equation}
	\vert L_{nz} \vert = \vert G_{nz} \vert = 8
	\end{equation}
		
	\begin{equation}
	\omega_{i} = \gamma_{i} = 0.5 \ \forall \ \omega_{i} \in L_{nz}, \gamma_{i} \in G_{nz}
	\end{equation}

	Finally, there is overlap in the $x_{i}$ which have non-zero $\omega_{i}$ and non-zero $\gamma_{i}$. Specifically, six $x_{i}$ have non-zero values of both coefficients, and four have a non-zero value of just one of $\omega_{i}$ or $\gamma_{i}$.

	\item[2.] A sparse linear DGP with squared terms characterized by
	
	\begin{equation}
	Y_{i} = D_{i}\beta_{0} + X_{i} \boldsymbol{\gamma_{0}} + X_{i}^{2} \boldsymbol{\gamma_{1}} + \epsilon_{i} ; \ \epsilon_{i} \sim N(0,1)
	\end{equation}
	
	\begin{equation}
	D_{i} = Bernoulli \left( \frac{1}{1 + e^{-(X_{i} \boldsymbol{\omega_{0}} + X_{i}^{2} \boldsymbol{\omega_{1}})}} \right)
	\end{equation}

	Additionally define $L_{0_{nz}} = \{\omega_{0_{k}} \in \boldsymbol{\omega_{0}} \vert \omega_{0_{k}} \neq 0\}$, $G_{0_{nz}}= \{\gamma_{0_{k}} \in \boldsymbol{\gamma_{0}} \vert \gamma_{0_{k}} \neq 0\}$, $L_{1_{nz}} = \{\omega_{1_{k}} \in \boldsymbol{\omega_{1}} \vert \omega_{1_{k}} \neq 0\}$, and $G_{1_{nz}}= \{\gamma_{1_{k}} \in \boldsymbol{\gamma_{1}} \vert \gamma_{1_{k}} \neq 0\}$. Then,
	
	\begin{equation}
	\vert L_{0_{nz}}\vert = \vert G_{0_{nz}}\vert = 8;  \vert L_{1_{nz}}\vert = \vert G_{1_{nz}}\vert = 2
	\end{equation}
	
	\begin{equation}
	\omega_{i} = \gamma_{i} = 0.5 \ \forall \ \omega_{i} \in \{L_{0_{nz}}, L_{1_{nz}}\}, \gamma_{i} \in \{G_{0_{nz}}, G_{1_{nz}}\}
	\end{equation}
	
		There is the same level of overlap in the $x_{i}$ which have non-zero coefficient values of $\omega_{0_{i}}$ and non-zero $\gamma_{0_{i}}$ as in the linear sparse DGP. This need not be true of the squared terms, however.
		
		\item[3.] A random NN DGP is characterized by
		
		\begin{equation}
		Y_{i} \sim D_{i}\beta_{0} + F_{Y}(X_{i}) + \epsilon_{i} ; \ \epsilon_{i} \sim N(0,1)
		\end{equation}
		
		\begin{equation}
		D_{i} = Bernoulli \left(\frac{1}{1 + e^{-(F_{D}(X_{i}))}} \right)
		\end{equation}
		
		$F_{Y}(\cdot)$ and $F_{D}(\cdot)$ each have one very wide hidden layer and one shallow hidden layer. The first layer for each network utilizes an ELU activation function, while the second utilizes a ReLU activation function. Importantly, the ELU activation function is never used to learn the matching embedding. We deliberately use a different structure as well as the ELU activation functions in order to ensure that this DGP is not nested in the neural nets which will learn the matching embedding.
		
The network weights for both $F_{Y}(\cdot)$ and $F_{D}(\cdot)$, which can be seen as analogous to the coefficient values of first two DGPs above, are generated with some correlation in order to induce an upward bias of roughly the same magnitude (here, referring to the bias that results from simply running a linear regression of $y$ on treatment) as the simpler DGPs. Additionally, network weights are randomly set to $0$ with $p = 0.5$ in order to limit the complexity of the DGP.

\end{itemize}

\section{Description of variable selection methods for matching}

The below methods are an instantiation of the general procedure described in Section 3.1. Here, that framework is applied to variable selection algorithms where the learned representation of the data is a subset of the input features.

For both of the below methods, we recommend standard tuning of the penalty parameter through cross-validation.

\subsection{L1 Method}

First, we run a LASSO of $y$ on $X$ and find $\boldsymbol{\beta_{ly}}$:

\begin{equation}
\boldsymbol{\beta_{ly}} = \argminA_{\boldsymbol{\beta}} \sum_{i =1}^{n} (y_{i} - \mathbf{X_{i}' \boldsymbol{\beta}})^2 + \lambda \sum_{j = 1}^{p} \vert \beta_{j} \vert 
\end{equation}

Second, we run an L1-penalized logit of $d$ on $X$ and find $\boldsymbol{\beta_{ld}}$:

\begin{equation}
\boldsymbol{\beta_{ld}} = \argminA_{\boldsymbol{\beta}} -\log \(\prod_{i:d_{i} = 1} p(\mathbf{X_{i}})   \prod_{j:d_{j} = 0} (1 - p(\mathbf{X_{j}}))\)  + \lambda \sum_{k = 1}^{p} \vert \beta_{k} \vert 
\end{equation}

We then collect $\mathbf{X_{select}} = \{X_{i} \vert \beta_{ly_{i}} \neq 0 \} \cup \{X_{j} \vert \beta_{ld_{j}} \neq 0 \}$.

In the simulations we present, we select the penalty parameter $\lambda$ through k-fold cross-validation separately for each regression. However, it is not obvious that selecting the optimal $\lambda$ which arises from solving the outcome and treatment prediction problems will coincide with the optimal $\lambda$ for the downstream matching estimator. Specifically, the cost of selecting irrelevant variables may be much larger for a matching estimator due to the order of the non-parametric bias. In fact, our simulation results on a true sparse DGP suggest that, when the goal is to input selected variables into a matching estimator, a more aggressive penalty may be necessary than that which is selected through cross-validation on the prediction problems.

\subsection{RRF Method}

The regularization framework for tree-based methods amounts to comparing the (penalized) gain for a split on a new variable to the maximum gain possible from splitting on a variable that has already been used for splitting. Formally, at any given node, define the set of features used in previous splits as $F$. For some feature $X_{j} \notin F$, we would like to ensure that it is selected only if $gain(X_{j}) \gg \max_{ \{i \vert X_{i} \in F \}} gain(X_{i})$.

We can accomplish this through applying a penalty $\lambda \in \[0,1\]$ to $gain(X_{j})$ for all $X_{j} \notin F$. Define $gain_{R}(X_j)$ as:

\begin{equation}
gain_{R}(X_j) = \begin{cases} 
      \lambda \cdot gain(X_{j}) & X_{j} \notin F \\
      gain(X_{j}) & X_{j} \in F
   \end{cases}
\end{equation}

An RRF model is then trained on both the outcome and treatment prediction problems, and we collect $\mathbf{X_{select}} = \{X_{i} \vert X_{i} \in F_{y} \} \cup \{X_{j} \vert X_{j} \in F_{d} \}$, where $F_{y}$ is the set of features selected by the outcome model and $F_{d}$ is the set of features selected by the treatment model.

Similarly to the L1 methods, we select the regularization parameter $\lambda$ through k-fold cross validation. The simulation results suggest that, as is the case with the L1 methods, this process may select a too-conservative penalty parameter--in this case corresponding to a $\lambda$ which is \textit{larger} that what would be optimal if the non-parametric nature of the matching estimator were fully internalized when selecting $\lambda$.

In fact, in our simulations that boast a true sparse DGP, nearly every available covariate ends up in $\mathbf{X_{select}}$. However, we use variable importance, $c_{j}$, as measured by average informational gain, to scale each $X_{j} \in \mathbf{X_{select}}$, so that matching is performed on the space defined by the set $\mathbf{X_{scale}} = \{c_{j} \cdot X_{j} \vert X_{j} \in \mathbf{X_{select}} \}$. Scaling selected features according to their importance has the effect of greatly reducing their contribution to the distance metric employed in the matching step, and therefore mitigating the non-parametric bias induced by their inclusion. However, a more elegant solution should involve a modified algorithm for selecting $\lambda$, as well as careful tuning of other parameters which can affect the magnitude of $\mathbf{X_{select}}$, such as the number of trees, the number of splitting candidates to be evaluated at each node, and the depth of each tree.

\section{Description of IHDP settings}

The response surfaces we use follow Li and Fu (2017), specifically:

\begin{itemize}

\item[i.] $Y(0) = exp((X+W) \beta) + \epsilon_{0}$
\item[ii.] $Y(1) = X \beta - \alpha + \epsilon_{1}$
\item[iii.] The factual and counterfactual outcomes are assigned as in the standard Rubin causal model framework.

\end{itemize}

Where $W$ is an offset matrix with each element equal to $0.5$ and $\beta \in \mathbf{R^{d \times 1}}$ is a vector of coefficients with each element sampled randomly from $(0, 0.1, 0.2, 0.3, 0.4)$ with respective probabilities $(0.6, 0.1, 0.1, 0.1, 0.1)$. The elements of the vectors $\epsilon_{0}$ and $\epsilon_{1}$ are randomly sampled from the standard normal distribution $N(0,1)$, and $\alpha \in \mathbf{R^{n \times 1}}$ is a constant vector with the value of the constant chosen such that the ATT will be equal to 4. We simulate 50 such response surfaces and report the results in Section 5. Both NN and SNN result in small bias and RMSE, providing further support for the efficacy of these feature learning methods across a variety of settings. \newline

\section{Proof of Theorem 2.1}

\begin{theorem} [Consistency]
	\label{thm:consistency}
	Suppose a metric $d$ on the space of $\mathbb X$. Additionally,
	\begin{enumerate}[label=\alph*.]
		\item suppose $\exists \ C >0$ such that for every $x,y \in \mathbb{X}$ \textit{either}
		\begin{itemize}
			\item $PSM$: $d(x,y) \geq C \cdot |\rho(x) - \rho(y)|$
			\item $PGM_C$: $d(x,y) \geq C \cdot |m(x) - m(y)|$
		\end{itemize}
		then $\hat{\ATT}_{d}$ is an asymptotically unbiased estimator of $\ATT$.
		\item suppose $\exists \ C >0$ such that for every $x,y \in \mathbb{X}$ \textit{either}
		\begin{itemize}
			\item $PSM$: $d(x,y) \geq C \cdot |\rho(x) - \rho(y)|$
			\item $PGM_T$: $d(x,y) \geq C \cdot |(m(x) + \tau(x)) - (m(y) + \tau(y))|$
		\end{itemize}
		then $\hat{\ATUT}_{d}$ is an asymptotically unbiased estimator estimator of $\ATUT$.
	\end{enumerate}
	And consistency of either estimator would follow from application of a law of large numbers to the residuals.
\end{theorem}

\begin{proof}
First note that by combining Assumptions 1-3, we know the density on $\mathbb X$ is everywhere positive for both treated and control units. Then for the case of part (a), there must exist a sequence $\delta_n \to 0$ such that either 
\begin{align*}
\textrm{$PGM_T$: } |m(x) - m(x_{j^*_{d}(i)})| < \delta_n
\intertext{or}
\textrm{$PSM$: } |\rho(x) - \rho(x_{j^*_{d}(i)})| < \delta_n.
\end{align*}
Now, write the true ATT and corresponding estimator as
\begin{align*}
ATT &= \Tavg \left[Y_i - Y_i(0)\right]\\
\hat \ATT_{d} &= \Tavg \left[Y_i - Y_{j^*_{d}(i)}\right]\\
\intertext{and the expectation of difference}
E[\hat \ATT_{d} - ATT] &= E\left(\Tavg \left[Y_{j^*_{d}(i)} - Y_i(0)\right]\right)\\
&=\Tavg \left[m(X_i)-  m(X_{j^*_d(i)})\right]
\end{align*}

is satisfied. In the $PGM_T$ case, it is immediate that our bias is bounded by
	\begin{align*}
	E\left|\Tavg \left[m(X_i)-  m(X_{j^*_d(i)})\right] \right|&\leq \delta_n \to 0,
	\end{align*}
whereas in the $PSM$ case the bias is controlled by standard arguments about the estimated propensity score (see for example (Abadie, 2016)).

For (b), the proof has the same structure but in the $PGM_C$ case depends on our ability to provide matches with similar counterfactual \textit{treated} outcomes. That is,
\begin{align*}
E[\hat \ATUT_{d} - ATUT] &= E\left(\Cavg \left[Y_{j^*_{d}(i)} - Y_i(0)\right]\right)\\
&=\Cavg \left[(m(X_i) + \tau(X_i)) -  (m(X_{j^*_d(i)}) + \tau(X_{j^*_d(i)}))\right],
\end{align*}
which can be bounded so long as $|(m(X_i) + \tau(X_i)) -  (m(X_{j^*_d(i)}) + \tau(X_{j^*_d(i)}))| < \delta_n.$

This difference between the treated and control case illuminates the importance of learning a metric for ATT (ATUT) estimation on only outcome labels from control (treated) data.

\end{proof}

\section{Computing infrastructure}

All simulations and empirical demonstrations were executed in R. We use R wrappers for \textit{keras} and \textit{tensorflow} to build all neural networks, and the \textit{glmnet} and \textit{RRF} packages to estimate the variable selection methods described in Supplement B.

We have run the code on a Linux HPC as well as locally on Windows machines. The least powerful machine on which we've run the code is an 8-core, 2.90GHz Windows desktop with an Intel Core i7-7700T processor and 16GB of RAM.

\end{document}